\documentclass[a4paper]{scrartcl}
 
\usepackage{authblk}
\usepackage{stmaryrd}
\usepackage{amssymb, amsmath, amsthm}
\usepackage{mathtools,lmodern}
\usepackage{hyperref}
\usepackage{color}
\usepackage{changepage} 

\usepackage{paralist} 
\usepackage{tikz}

\newtheorem{theorem}{Theorem}
\newtheorem{definition}{Definition}
\newtheorem{example}{Example}
\newtheorem{corollary}{Corollary}
\newtheorem{lemma}{Lemma}
\newtheorem{remark}{Remark}

\newcommand{\ignore}[1]{}
\definecolor{maroon}{RGB}{128,0,0}
\definecolor{lava}{RGB}{207, 16, 32}

\definecolor{forestgreen}{RGB}{34,139,34}

\definecolor{midnightblue}{RGB}{25,25,112}
\definecolor{hanblue}{RGB}{82, 24, 250}

\definecolor{indigo}{RGB}{75,0,130}
\definecolor{hanpurple}{RGB}{82, 24, 250}

\definecolor{teal}{RGB}{0, 128, 128}

\definecolor{rawumber}{RGB}{130, 102, 68}

\newcommand{\tcr}[1]{\textcolor{lava}{#1}}

\newcommand{\tcg}[1]{\textcolor{forestgreen}{#1}}

\newcommand{\tcb}[1]{\textcolor{blue}{#1}}

\newcommand{\bl}[2]{#1^{(#2)}}

\newcommand{\no}{\sqsubset}

\newcommand{\noq}{\sqsubseteq}
\newcommand{\noqr}{\sqsupseteq}
\newcommand{\lcp}{\sqcap}
\newcommand{\biglcp}{\bigsqcap}
\newcommand{\ry}{\reflectbox{\ensuremath{y}}}

\newcommand{\bx}{\bar{x}}
\newcommand{\by}{\bar{y}}
\newcommand{\bz}{\bar{z}}

\newcommand{\px}{(x,\bx)}
\newcommand{\py}{(y,\by)}
\newcommand{\pz}{(z,\bz)}

\newcommand{\hq}{\hat{q}}

\newcommand{\dq}{\dot{q}}

\newcommand{\ie}{\quad\text{ i.e. }\quad}

\newcommand{\al}{\Sigma}
\newcommand{\als}{\al^\ast}
\newcommand{\ct}{\mathsf{C}_{\al}}

\newcommand{\ew}{\varepsilon}

\newcommand{\prf}{\textsf{prf}}

\newcommand{\abs}[1]{\left|#1\right|}

\newcommand{\N}{\mathbb{N}}

\title{Computing the longest common prefix of a context-free language in polynomial time}
\author{Michael Luttenberger}
\author{Raphaela Palenta}
\author{Helmut Seidl}
\affil{TU M\"unchen, Germany,
  \texttt{\{luttenbe,palenta,seidl\}@in.tum.de}}

\newcommand{\txtLCP}{\textsf{lcp}\ }

\begin{document}

\maketitle

\begin{abstract}
We present two structural results concerning the longest common prefixes of non-empty languages.
First, we show that the longest common prefix of the language generated by a context-free grammar of size $N$
equals the longest common prefix of the same grammar where the heights of the derivation trees are bounded by
$4N$.
Second, we show that each non-empty language $L$ has a representative subset of at most three elements which behaves 
like $L$ w.r.t.\ the longest common prefix as well as w.r.t.\ longest common prefixes of $L$ after unions or
concatenations with arbitrary other languages.
From that, we conclude 
that the longest common prefix, and thus the longest common suffix, of a context-free language can be computed in polynomial time.
 \end{abstract}

\section{Introduction}
Let $\al$ denote an alphabet. On the set $\als$ of all words over $\al$, the prefix relation provides us with a partial ordering
$\sqsubseteq$ defined by $u\sqsubseteq v$ iff $u u' = v$ for some $u'\in\als$.
The {\em longest common prefix} (\txtLCP for short) of a non-empty set $L\subseteq\als$ then is given by the greatest lower bound $\biglcp L$ of $L$
w.r.t.\ this ordering. 
For two words $u,v\in \als$, we also denote this greatest lower bound as $u\lcp v$.
\ignore{
 denote their longest common prefix, i.e.\ the longest word of $\prf(u)\cap \prf(v)$.
The longest prefix common to all words in a non-empty language $L\subseq\als$ is given by 
$\biglcp L = \biglcp_{u\in L} u$ of $L$, i.e.\ the longest word of $\bigcap_{u\in L} \prf(u)$.
}
Our goal is to compute the \txtLCP when the language $L$ is context-free, i.e., generated by
a context-free grammar (CFG) --- we therefore assume wlog.\ that $\al$ contains at least two letters.

The computation of the \txtLCP 
(sometimes also {\em maximum common prefix}) is well studied for finite languages, in particular in the setting of string matching based on suffix arrays (e.g., \cite{KasaiLAAP01}) where the string is given explicitly. Very often, 
strings can be efficiently compressed using {\em straight-line programs (SLPs)} --- essentially CFGs which produce exactly one word. Interestingly, many of the standard string operations can still be done efficiently also on SLP-compressed strings (see, e.g., \cite{Lohrey2012}). As the union of SLPs is a (acyclic) CFG, the question of computing the lcp of a context-free language naturally arises. 
CFGs also represent a popular formalism to specify sets of well-formed words.
Assume that we are given a CFG for the legal outputs of a program.
This CFG might be derived from the specification as well as from an abstract interpretation of the program.
Then the \txtLCP of this language represents a prefix which can be output already,
before the program actually has been run.
This kind of information is crucial for the construction of \emph{normal forms}, e.g., of 
string producing processors such as
linear tree-to-string transducers \cite{Boiret2016,Laurence2011}. For these devices, the normal forms 
have further interesting applications as they allow for simple algorithms to decide equivalence \cite{Palenta2016}
and enable efficient learning \cite{Laurence2014}.

Obviously, the \txtLCP of the context-free language $L$ is a prefix of the shortest word in $L$. 
Since the shortest word of a context-free language can be effectively computed, 
the \txtLCP of $L$ is also effectively computable. 
The shortest word generated from a context-free grammar $G$, however, may be of length
exponential in the size of $G$. Therefore, it is an intriguing question whether or not the
\txtLCP can be efficiently computed.
Here, we show that the longest common prefix can in fact be computed in polynomial time.
As the words the algorithm computes with may be of exponential length, we have to resort
to \emph{compressed} representations of long words by means of SLPs 
\cite{Plandowski}.
We will rely on algorithms for basic computational problems for SLPs as presented, 
e.g., in \cite{Lohrey2012}.

Our method of computing $\biglcp L$ is based on two structural results.
First we show in 
Section~\ref{sec_height}
that it suffices to consider the finite sublanguage 
of $L$ consisting of those words, for which there is a derivation tree of height at most $4N$ ---
with $N$ the number of nonterminals for a CFG of $L$.\footnote{To simplify the presentation we assume that the CFG is proper, i.e.\ we will rule out production rules of the form $A\to B$ and $A\to \ew$ (with $A,B$ nonterminals and $\ew$ the empty word).}
This implies that (1) in the proof of our main result we can replace the grammar by an \emph{acyclic} context-free grammar, and (2) the actual fixpoint iteration to compute the \txtLCP will converge within at most $4N$ iterations.
Second we show in Section~\ref{sec_rep} that for every non-empty language $L$ there is a subset $L'\subseteq L$ 
of at most three elements which is \emph{equivalent} to $L$ w.r.t.\ the \txtLCP after 
arbitrary concatenations with other words.
This means that for every word $w$, the language $L'w$ has the same \txtLCP as $Lw$. 

We illustrate both results by examples.
\iftrue
For the first result, i.e.\ the restriction to derivation trees of bounded height, consider the language 
\[
L:=\{a^2b(\tcr{a^2b})^i\tcb{a^2b}(\tcg{a^2ba})^ia^2ba^2ba^3\mid i \in \N_0\}
\] 
generated by the context-free grammar consisting of the following rules over the alphabet $\al=\{a,b,c\}$ and the six nonterminals $\{S,X,A_2,A_1,X_2,X_1\}$:
\[
\begin{array}{llllll}
S\to X_2A_2bA_2bA_2a & A_2 \to aA_1 & A_{1} \to a & X\to A_2b \\
 & X_2\to aX_1 & X_{1} \to abX & X\to X_2 A_2 ba
\end{array}
\]
It is easy to check that here the \txtLCP is already determined by repeating the derivation of $X$ to $\tcb{aab}X\tcb{aaba}$ at most two times, which corresponds to the sublanguage consisting of all words which have a derivation tree of height at most $9$.
\[
\begin{array}{lccll}
\biglcp L
& = &      & aab\tcb{aab}aabaabaa\, a & (i=0)\\
&   & \lcp & aab\tcr{aab}\tcb{aab}\tcg{aaba}a\, abaabaaa & (i=1)\\
&   & \lcp & aab\tcr{aabaab}\tcb{aab}\tcg{aa\, ba}\tcg{aaba}aabaabaaa & (i=2)\\
&   & \lcp & aab\tcr{aabaabaab}\tcb{aa\, b}\tcg{aabaaabaaaba}aabaabaaa & (i= 3)\\
&   & \lcp & aab\tcr{aabaabaabaa\, b}\ldots & (i\ge 4)\\
& = &      & aabaabaabaabaa\\
\end{array}
\]
We remark that the bound of $4N$, i.e.\ $24$ for this example, on the height resp.\ the number of iterations needed to converge is a crude overapproximation based on the pigeon-hole principle which does not take into account the structure of the grammar. The actual computation of the \txtLCP may thus terminate much earlier, in particular when taking the dependency of nonterminals into account as done in Example~\ref{ex:alg}.
\else
For the first result, i.e.\ the restriction to derivation trees of finite height, fix any $k\in\N_0$ and consider the language $L:=\{a^{2^k}b(\tcr{a^{2^k}b})^i\tcg{a^{2^k}b}(\tcb{a^{2^k}ba})^ia^{2^k}ba^{2^k}ba^{2^k}a\mid i \in \N_0\}$ generated by the context-free grammar consisting of the following rules over the alphabet $\al=\{a,b,c\}$:
\[
\begin{array}{c}
S\to X_kA_kbA_kbA_ka \quad X\to X_kA_kba \mid A_kb \\
A_k \to aA_{k-1}\quad \ldots \quad A_{1} \to a\\
X_k \to aX_{k-1}\quad \ldots \quad X_{1} \to abX \\
\end{array}
\]
Enumerating the words of $L$ for increasing $i$ 
\[
\begin{array}{l|l}
i & a^{2^k}b(\tcr{a^{2^k}b})^i \tcg{a^{2^k}b}\tcb{a^{2^k}ba})^ia^{2^k}ba^{2^k}ba^{2^k}a\\
\hline
0 & a^{2^k}b\tcg{a^{2^k}b}a^{2^k}ba^{2^k}ba^{2^k}a\\
1 & a^{2^k}b\tcr{a^{2^k}b}\tcg{a^{2^k}b}\tcb{a^{2^k}ba}{a^{2^k}ba^{2^k}ba^{2^k}a}\\
2 & a^{2^k}b\tcr{a^{2^k}b}\tcr{a^{2^k}b}\tcg{a^{2^k}b}\tcb{a^{2^k}ba}\tcb{a^{2^k}ba}{a^{2^k}ba^{2^k}ba^{2^k}a}\\
3 & a^{2^k}b\tcr{a^{2^k}b}\tcr{a^{2^k}b}\tcr{a^{2^k}b}\tcg{a^{2^k}b}\tcb{a^{2^k}ba}\tcb{a^{2^k}ba}\tcb{a^{2^k}ba}{a^{2^k}ba^{2^k}ba^{2^k}a}\\
\vdots & \vdots\\
\end{array}
\]
we obtain that 
\[
\begin{array}{lcl}
\biglcp L & = & a^{2^k}b\tcg{a^{2^k}b}a^{2^k} a^{2^k}ba^{2^k}ba^{2^k} \\
& = & a^{2^k}b\tcg{a^{2^k}b}a^{2^k}ba^{2^k}ba^{2^k}a \lcp a^{2^k}b\tcr{a^{2^k}b}\tcr{a^{2^k}b}\tcg{a^{2^k}b}\tcb{a^{2^k}ba}\tcb{a^{2^k}ba}{a^{2^k}ba^{2^k}ba^{2^k}a}\\
& = & \biglcp \{a^{2^k}b(\tcr{a^{2^k}b})^i\tcg{a^{2^k}b}(\tcb{a^{2^k}ba})^ia^{2^k}ba^{2^k}ba^{2^k}a \mid 0 \le i \le 2\}
\end{array}
\]
By construction, the $i$-th word of $L$ has a derivation tree of height $(i+2)(k+1)+1$ (excluding the leaves representing the letters).
\tcr{Raus oder noch genauer, um die Schranke von $4N=4(2+3(k+1))=8+12(k+1)$ besser zu motivieren?}
Note that the grammar in this example is essentially a linear grammar extended with SLPs in order to encode the exponentially long words $a^{2^k}$. We show in Section~\ref{sec_height} that in order to bound the height of the derivation trees it indeed suffices to consider essentially linear languages.
\fi

In order to compute the \txtLCP recursively, we call two languages $L_1,L_2\subseteq \Sigma^\ast$ {\em equivalent w.r.t.\ the \txtLCP} 
if  for all words $w \in \Sigma^*$ we have that $\biglcp(L_1w) = \biglcp(L_2w)$.
In Section~\ref{sec_rep} we show that every language $L$ can be reduced to a sublanguage $L'$ consisting of at most three words so that $L$ and $L'$ are equivalent w.r.t.\ the \txtLCP.
In fact, this result can be motivated by considering the special case of a language of the form $L=\{u,uv_1\}$ (with $u,v_1\in\als$) where we have $\biglcp (Lw) = u(w\lcp v_1^\omega)$ for any $w\in\als$ (see also Section~\ref{sec_rep}). From this observation one immediately obtains that for finite languages $L'=\{uv_1,uv_2,\ldots,uv_k\}$ we have $\biglcp(L'w) = u(w\lcp v_1^\omega\lcp v_2^\omega\lcp \ldots \lcp v_k^\omega)$ and that one only needs to keep those two $uv_i,uv_j$ for which $v_i^\omega\lcp v_j^\omega$ is minimal. The result then extends to arbitrary languages.
E.g., in case of the language $L = a (ba)^*$ we only need the sublanguage $\{a,aba\}$ (with $\ew^\omega \lcp (ba)^\omega := (ba)^\omega$) as the words $a$ and $aba$ suffice to characterize both $\biglcp L=a$ and the period $ba$ that generates all suffices. For comparison, in case of $L = abab+aba(ba)^\ast$ the \txtLCP is $aba$, which can only be extended to at most $abab = aba(b^\omega \lcp (ba)^\omega)$. We therefore need to remember $\{aba,abab,ababa\}$: 
the sublanguages $\{aba,abab\}$ resp.\ $\{aba,ababa\}$ preserve $\biglcp L = aba$ but can be extended by $b^\omega$ resp.\ $(ba)^\omega$; whereas $\{abab,ababa\}$ only captures the maximal extension of $\biglcp L$, but does not preserve $\biglcp L$ itself.

In order to compute the \txtLCP of a given context-free language $L$ we then (implicitly) unfold the given context-free
grammar into an acyclic grammar, and compute for every nonterminal of the unfolded grammar an equivalent sublanguage
of at most three words, each compressed by means of a SLP, instead of the actual language. From this finite representation of $L$ we then can easily obtain its \txtLCP.
Altogether, we arrive at a polynomial time algorithm.

\section{Preliminaries}

$\al$ denotes a (finite) alphabet. We assume that $\al$ contains at least two letters as any context-free language over a unary alphabet is regular. $\als$ is the set of all finite words over $\al$ with $\ew$ the empty word, $\al^\omega$ the set of all (countably) infinite words over $\al$. We use ($\omega$-)rational expressions to denote words and languages, e.g.\ $w^\ast = \ew + w + ww + \ldots = \sum_{i\in\N_0} w^i$ and $w^\omega = wwwwwwwwwww\ldots$.

By $\ct = \{(u,v)\in \als\times\als\}$ we denote the set of all pairs of finite words over $\al$. We define a multiplication on $\ct$ by $(x,\bx)(y,\by):=(xy,\by\bx)$. For $(x,\bx)\in\ct$ and $w\in \als$ set $(x,\bx)w= xw\bx$. As in the case of words, we set $(x,\bx)^0:=(\ew,\ew)$, $(x,\bx)^{k+1}:=  (x,\bx)(x,\bx)^k$ and $(x,\bx)^\ast := \sum_{k\ge 0} (x,\bx)^k$ for all $x,\bx\in\als$ and $k\in\N_0$.

Note that we slightly deviate from standard notation when it comes to the prefix order (i.e.\ $u<w$) and the common prefix (i.e.\ $u\wedge v$) of two words in order to avoid the clash with the notation for conjunction ($\wedge$):
For $u,v\in \als$ we write $u\noq v$ ($u\no v$) to denote that $u$ is a (strict) prefix of $v$, i.e.\ $v=uw$ for some $w\in\als$ ($w\in\al^+$). For $L\subseteq \als$ (with $L\neq \emptyset$) its {\em longest common prefix} (\txtLCP) $\biglcp L$ is given by the greatest lower bound of $L$
w.r.t.\ this ordering.
We simply write $u\lcp v$ for $\biglcp\{u,v\}$. Note that for any word $w\in L$ there is at least one word $\alpha\in L$ s.t.\ $\biglcp L = w\lcp \alpha$; we call any such $\alpha$ a {\em witness (w.r.t.\ $w$)}. Note that $\lcp$ is commutative and associative;  concatenation distributes from the left over the \txtLCP (i.e.\ $u(v\lcp w) = uv\lcp uw$); and the \txtLCP is monotonically decreasing on the union of languages, i.e.\ $\biglcp (L\cup L') = (\biglcp L)\lcp (\biglcp L')$. The \txtLCP of infinite words is defined analogously.

A word $p\in \als$ is called a power of a word $q$ if $p \in q^\ast$; then $q$ is called a root of $p$; if $p\neq \ew$ is its own shortest root, $p$ it is called {\em primitive}. Two words $u,v$ are {\em conjugates} if the is a factorization $u=pq$ and $v=qp$. We recall two well-known results:

\begin{lemma}[Commutative Words, \cite{Choffrut1997}]\label{lem:com}
Let $u,v\in\als$ be two words.
If $uv=vu$, then $u,v\in p^\ast$ for some primitive $p\in \als$.
\end{lemma}

\begin{lemma}[Periodicity Lemma of Fine and Wilf, \cite{FineWilf1965}]\label{lem:fw}
Let $u,v\in\al^+$ be two non-empty words. 
If $\abs{u^\omega \lcp v^\omega}\ge \abs{u}+\abs{v}-\gcd(\abs{u},\abs{v})$, then $uv=vu$.
\end{lemma}

Combining these two lemmata yields the following result which is a useful tool in the proofs to follow (see also lemma 3.1 in \cite{Choffrut1997} for a more general version of this result):
\begin{corollary}\label{cor:fw}
Let $u,v\in\als$ with $uv\neq vu$. 

Then $u^\omega\lcp v^\omega = uv\lcp vu$ with $\abs{uv\lcp vu}< \abs{u}+\abs{v}-\gcd(\abs{u},\abs{v})$.
\end{corollary}
\begin{proof}
Since the bound of the size of $\abs{uv\lcp vu}$ follows from Lemma \ref{lem:fw}
we only have to show that $uv\lcp vu = u^\omega\lcp v^\omega$.
If $\abs{u} = \abs{v}$, then $uv\neq vu$ implies $u\neq v$ and $uv\lcp vu = u\lcp v = u^\omega \lcp v^\omega$.

W.l.o.g.\ we assume that $\abs{u}<\abs{v}$. As $uv\neq vu$, we have $\ew \neq u$.
Let $v\lcp u^\omega = u^ku'\no u^{k+1}$ with $v=u^ku'v'$ and $u=u'u''$.
It follows that
$uv\lcp vu = uu^ku'v' \lcp u^ku'v'u = u^k(uu'v'\lcp u'v'u) = u^k u'(u''u'v'\lcp v'u'u'')$.

If $v'\neq \ew$, we have $u''u'v'\lcp v'u'u'' = u'' \lcp v' = \ew$, and thus
$uv\lcp vu = u^ku' = v\lcp u^\omega = v^\omega \lcp u^\omega$.

So assume $v'=\ew$, i.e.\ $v\no u^\omega$ with $k>0$ as $\abs{u}<\abs{v}$. As $uv=u^ku'u''u' \neq u^ku'u'u'' = vu$, also $u'u''\neq u''u'$.
Hence $uv\lcp vu = u^k u'(u''u'\lcp u'u'') = u^{k+1}u \lcp vv = u^\omega \lcp v^\omega$,
which concludes the proof.
\end{proof}
Here is a short example for the last corollary:
\begin{example}
Let $u=aab$, $v=aaba =ua$. Then $uv\lcp vu = aabaaba\lcp aabaaab = aabaa = va$ and
$u^\omega\lcp v^\omega = aabaabaabu^\omega \lcp aabaaabav^\omega = aabaa$
with $\abs{aabaa}=\abs{u}+\abs{v} - \gcd(\abs{u},\abs{v})-1$. I.e.\ the bound is sharp.
Note that this example also shows, that even if $uv \neq vu$ and $\ew \neq u \no v$, we still can have $v \no uv\lcp vu$.
\end{example}
We briefly discuss properties of the \txtLCP for very simple regular languages. These will be used several times in the proofs of Section~\ref{sec_height} in order to bound the height of the derivation trees we need to consider:
\begin{lemma}\label{lem:lcp-reg-1}
Let $y\neq \ew$, then $w\lcp yw = w\lcp y^i w = \biglcp y^\ast w = w\lcp y^\omega$ for all $i>0$.
\end{lemma}
\begin{proof}
Let $w\lcp y^\omega = y^{k}y'\no y^{k+1}$ with $w=y^ky'w'$. Then for any $i>0$ we have 
$w\lcp y^iw = w \lcp y^{k+i}y'w' = w\lcp y^\omega$ where the last equality holds as $i>0$ and $w\lcp y^{k+1}=w\lcp y^\omega\no y^{k+1}$.
\end{proof}
\begin{lemma}\label{lem:w}
If $w\not\noq yw$, then $\biglcp y^\ast w = w\lcp y^iw \no w$ for all $i>0$.
\end{lemma}
\begin{proof}
Since $w\not\noq yw$, we have $w\neq\ew$ and $y\neq \ew$. By Lemma~\ref{lem:lcp-reg-1} we thus have $\biglcp y^\ast w = w\lcp y^iw$ for any $i>0$, in particular for $i=1$. Define $w=y^ky'w'$ as in Lemma~\ref{lem:lcp-reg-1}. As $w\not\noq yw$, we have $w'\neq \ew$ and thus $w\lcp yw = y^ky' \no w$.
\end{proof}
We assume that the reader is familiar with context-free grammars (CFGs). We briefly introduce the notation we use for CFGs in the following. A context-free grammar $G$ is given by a tuple $G=(\al,V,P,S)$ where $\al$ is the alphabet of terminals, $V$ is the set of nonterminals (also: variables), $P\subseteq V\times (V\cup\al)^\ast$ is the set of production rules where a rule $p=(A,\gamma)\in P$ is also written as $A\to \gamma$, and $S$ the axiom. The language generated by $G$ is denoted by $L(G)$. $G$ is {\em proper} if $A\to\ew \not\in P$ and $A\to B\not\in P$ for all $A,B\in V$; $G$ is in {\em Chomsky normal form (CNF)} if all rules are of the form $A\to a\in V\times \al$ or $A\to BC\in V\to VV$.  For every CFG $G$ a proper CFG resp.\ a CFG in CNF  $G'$ can be constructed in time polynomial in the size of $G$ such that $L(G)\setminus\{\ew\}=L(G')$~\cite{DBLP:journals/didactica/LangeL09}.
As $\ew\stackrel{?}{\in} L(G)$ is decidable in time polynomial in the size of $G$, and trivially $\biglcp L = \ew$ if $\ew\in L$, we will assume that $\ew\not\in L(G)$ and that $G$ is proper from here on. For some proofs we assume in fact that $G$ is in CNF but only in order to simplify notation.

\section{LCP of a context-free language}\label{sec_height}

Our main result in this section, Theorem~\ref{thm:witness-height}, is that for every context-free language $L=L(G)$ generated by the given CFG $G$ its \txtLCP $\biglcp L$ is equal to the \txtLCP of its {\em finite} sublanguage $L'$ which contains only the words $w\in L$ which possess a derivation tree w.r.t.\ $G$ whose height (considering only nonterminals) is at most four times the number of nonterminals of $G$. 
For the main result we require the following technical theorem (see the following example).
\begin{theorem}\label{thm:lcp-pump}
Let $L=(x,\bx)[(y_1,\by_1)+\ldots+(y_l,\by_l)]^\ast w$ for $(x,\bx),(y_1,\by_1),\ldots,(y_l,\by_l)\in\ct$ and $w\in \als$.
Then: $$\biglcp L = \biglcp (x,\bx)[(y_1,\by_1)^{\le 2}+\ldots+(y_k,\by_l)^{\le 2}] w$$
Furthermore, if $\biglcp L = xw\bx\lcp xy^2w\by^2\bx\no xw\bx\lcp xyw\by\bx$ for some $(y,\by)\in\{(y_1,\by_1),\ldots,(y_l,\by_l)\}$, 
then w.r.t.\ this $y$ there exists some primitive $q\in\als$ and some $k>0$ such that
\[
yw=wq^k \wedge 
q\by \neq \by q \wedge
\biglcp L = xw\bar{x} \lcp xywq\bar{y}\bar{x} 
\wedge xwq^k(\by\lcp q^\omega)\noq \biglcp L \no xw q^{k+1} (\by\lcp q^\omega)
\]
\end{theorem}
The proof of the main theorem of this section, Theorem~\ref{thm:witness-height}, crucially depends on the observation that in the case $\biglcp L \no xw\bx\lcp xyw\by\bx$, all the words $y_i$ are powers of the same primitive word $p$ with $pw=wq$ and all that is needed to obtain a witness is one additional power of $p$ resp.\ its conjugate $q$ (with $pw=wq$) to which Theorem~\ref{thm:lcp-pump} refers to. 
We give an example in order to clarify the statement of Theorem~\ref{thm:lcp-pump} in the case of $l=2\wedge y_1y_2=y_2y_1$ which is central to Theorem~\ref{thm:witness-height}:
\begin{example}
We write $(y,\by)$ for $(y_1,\by_1)$ and $(z,\bz)$ for $(y_2,\by_2)$, respectively.
Let $(x,\bx)=(\ew,ababaaa)=(\ew,qqaaa)$, $(y,\by)=(\tcr{ab},\tcb{abaab})=(q,qaab)$, $(z,\bz)=(\tcr{ab},\tcg{abaac})=(q,qaac)$, and $w=\ew$ with $q=ab=y=z$. We then have:
\[
\begin{array}{lcl}
xw\bar{x} & = & ababaaa\\
xyw\bar{y}\bar{x} & = & \tcr{ab}\tcb{abaab}ababaaa \\
xzw\bz\bx & = &  \tcr{ab}\tcg{abaac}ababaaa\\
xyyw\by\bar{y}\bar{y}\bar{x} & = & \tcr{abab}\tcb{abaababaab}ababaaa \\
xyzw\bz\by\bx & = & \tcr{abab}\tcg{abaac}\tcb{abaab}ababaaa\\
xzyw\by\bz\bx & = & \tcr{abab}\tcb{abaab}\tcg{abaac}ababaaa\\
xzzw\bz\bz\bx & = & \tcr{abab}\tcg{abaacabaac}ababaaa\\
x(y+z)^{\ge 3}\ldots & = & \tcr{ababab}\ldots\\
xywq\by\bx & = & \tcr{ab}ab\tcb{abaab}ababaaa \\
xzwq\bz\bx & = & \tcr{ab}ab\tcg{abaac}ababaaa\\
\hline
\biglcp L & = & ababa\\
\end{array}
\]
So in this example, any word except for $xyw\by\bx$ and $xzw\bz\bx$ is a witness for the \txtLCP w.r.t.\ $xw\bx$. W.r.t.\ the proof of Theorem~\ref{thm:witness-height} it is important that also in general we can pick a witness which either is derived using only $(y,\by)$ or $(z,\bz)$ but not both, and that we need to use $(y,\by)$ resp.\ $(z,\bz)$ at most twice in order to get one additional copy of the conjugate $q$ of the primitive root of both $y$ and $z$.
\end{example}
To give an impression of the proof of Theorem~\ref{thm:lcp-pump} we show the case $l=1$.
The complete proof of Theorem~\ref{thm:lcp-pump} can be found in the appendix of \cite{DBLP:journals/corr/LuttenbergerPS17}.
\begin{lemma}\label{ssec:lcp-one-pt}
Let $L=(x,\bar{x})(y,\bar{y})^\ast w$. Then: $\biglcp L = \biglcp (x,\bar{x})(y,\bar{y})^{\le 2}w$.
\noindent 
If $\biglcp L \no xw\bx\lcp xyw\by\bx$,
then there is some primitive $q$ and some $k>0$ s.t.\
\[
yw=wq^k \wedge 
q\by \neq \by q \wedge
\biglcp L = xw\bar{x} \lcp xywq\bar{y}\bar{x} 
\wedge xwq^k(\by\lcp q^\omega)\noq \biglcp L \no xw q^{k+1} (\by\lcp q^\omega)
\]
\end{lemma}
\begin{proof}
Recall that for any $z\in L$ there is some {\em witness} $z'\in L$ s.t.\ $\biglcp L = z\lcp z'$. 
Our main goal is to show that w.r.t.\ $xw\bx$ we find a witness within $\{xy^iw\by^i\bx \mid i=0,1,2\}$.
What makes the proof technically more involved is that for Theorem~\ref{thm:witness-height} we need a stronger characterization of the case when $xyyw\by\by\bx$ is the only witness in this set.

If $y=\ew\vee \by=\ew$, then $L$ is actually regular and Lemma~\ref{lem:lcp-reg-1} already tells us that $xyw\by\bx$ is a witness (w.r.t.\ $xw\bx$). So wlog.\ $y\neq\ew\neq\by$.
If $w\not\noq yw$, then $\biglcp y^\ast w = w\lcp yw \no w$ by Lemma~\ref{lem:w} and thus $\biglcp L = x(w\lcp yw)$, i.e.\ $xyw\by\bx$ is again a witness.  

From now on we assume that $w\noq yw$. Then there is some conjugate $\ry$ of $y$ defined by $w\ry = yw$, and $xw$ is a prefix of $\biglcp L$ as
$xy^iw\by^i \bx = xw\ry^i\by^i\bx$. Wlog.\ we therefore assume $xw=\ew$ from now on so that $L$ becomes $\{y^i\by^i\bx\mid i\in \N_0\}$. 

Let $q$ be the primitive root of $y$ s.t.\ $y = q^k$ for a suitable $k>0$ (as $y\neq\ew$). By choosing $j> \abs{\bx}/\abs{y}$ we obtain $\biglcp L \noq \bx \lcp y^j \by^j\bx = \bx \lcp q^{kj} \no q^\omega$,
i.e.\ $\biglcp L \no q^\omega$.
We therefore factorize $\bx$ and $\by$ w.r.t.\ $q^\omega$:
Let $\bx = q^n q' \bx'$ with $\bx\lcp q^\omega = q^n q' \no q^{n+1}$; and let $\by = q^{k'} \hq \by'$ with $\by \lcp q^\omega = q^{k'}\hq \no q^{k'+1}$.
The words of $L$ have thus the form
$y^i\by^i\bx = q^{ik}\Bigl(q^{k'}\hq \by'\Bigr)^i q^n q' \bx'$.

If $q$ (resp.\ $y$) and $\by$ commute, then $\by=q^{k'}$ by Lemma~\ref{lem:com} (as $q$ is primitive) for some suitable $k'\in \N$. Then $L= (y\by)^\ast \bx = (q^{k+k'})^\ast q^nq'\bx'$ with $\biglcp L = q^nq'$, and $y\by\bx$ is again a witness w.r.t.\ $\bx$.
We thus also assume $q\by \neq \by q$ from here on.

If $q^{n}q' \noq q^{k+k'}\hq$, then $\biglcp L \noq q^{n}q'$ and $qy\by\bx$ is a witness w.r.t.\ $\bx$: by choice of $n$ we have $\bx\lcp q^\omega = \bx \lcp q^{n+1}$, by $q^nq'\noq q^{k+k'}\hq$ we also have $q^{n+1}\noq q^{k+k'+1}$; from this we obtain $\bx\lcp qy\by\bx = \bx \lcp q^{k+k'+1} \hq \bx = \bx \lcp q^{n+1} = q^nq'$. Thus, also $yy\by\by\bx$ is a witness w.r.t\ $\bx$.

Assume now that $q^{k+k'}\hq \no q^nq'$ and thus $q^{k+k'}\hq \noq \biglcp L$.
If $\biglcp L=q^{k+k'}\hq$, then $\bx\lcp y\by\bx = q^{k+k'}\hq$ has to hold, i.e.\ $y\by\bx$ has to be a witness. 
Thus assume $q^{k+k'}\hq \no \biglcp L$.
If $\by'\neq \ew$, then, as $q^{k+k'}\hq\no q^nq'$, we have that $q^nq'\lcp q^{k+k'}\hq\by' = q^{k+k'}\hq$ so that $y\by\bx$ is again a witness.
Hence assume $\by'=\ew$ resp.\ $\by=q^{k'}\hq$ for the remaining.
As $q$ and $\by$ do not commute, also $q$ and $\hq$ do not commute implying $q\hq\no \hq\ q \no q\hq$. 
Thus
\[
\begin{array}{rcl}
q^{k+k'}\hq \no \biglcp L \noq y\by\bx \lcp yy\by\by\bx & = &q^{k+k'}(\hq q^nq'\bx'\lcp q^k\hq \by\bx)\\[1mm]
 &\stackrel{n\ge k>0\wedge \hq \no q}{=}& q^{k+k'}(\hq q\lcp q\hq) \no q^{k+k'}q\hq
\end{array}
\]
That is either $y\by\bx$ or $yy\by\by\bx$ has to be a witness w.r.t.\ $\bx$ as $\biglcp L \no q^\omega$ and as we can extend $q^{k+k'}\hq$ by at most $\abs{q}-1$ symbols, i.e.\ we need at most one additional copy of $q$ which is again given by $yyw\by\by\bx$ as $k>0$.
In particular, we have again that, if $yy\by\by\bx$ is a witness, then so is $qy\by\bx$.
\end{proof}
\noindent
Using Theorem~\ref{thm:lcp-pump}, we now can show that we only need to consider a finite sublanguage of $L$ instead of $L$ itself:
\begin{theorem}\label{thm:witness-height}
Let $L=L(G)$ be given by a proper CFG $G=(\al,V,P,S)$.
Let $\hat{L}\subseteq L$ be the finite language of all words of $L$ for which there is a derivation tree w.r.t.\ $G$ of height\footnote{We measure the height of a derivation tree only w.r.t.\ nonterminals along a path from the root to a leaf.} at most $4N$ with $N=\abs{V}$.
Then: $\biglcp L = \biglcp \hat{L}$.
\end{theorem}
\begin{proof}
Let $N$ be the number of nonterminals of $G$.
Let $\sigma\in L$ be a shortest word, and $\alpha\in L$ a shortest word with $\biglcp L = \sigma \lcp\alpha$.
Set $\pi:=\biglcp L$.

We claim that there is at least one such $\alpha$ (for any fixed $\sigma$) that has an derivation tree w.r.t.\ $G$ of height less than $4N$.If $\sigma=\alpha$, we are done as $\sigma$ has a derivation tree of height less than $N$.
So assume $\sigma\neq \alpha$ s.t.\ $\sigma = \pi a \sigma'$ and $\alpha = \pi b\alpha'$ with $a\neq b$ and $a,b\in\al$. Then fix any derivation tree $t$ of $\alpha$ w.r.t.\ $G$.

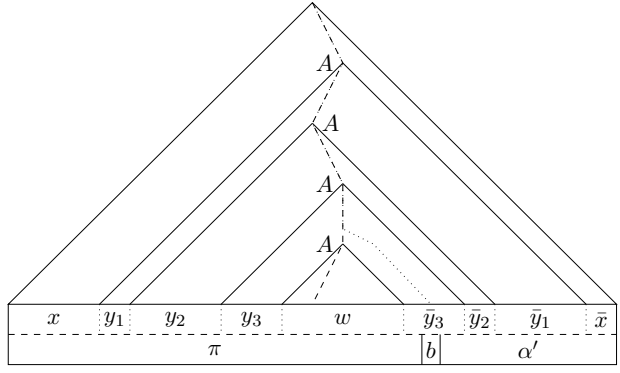
\begin{figure}
\begin{center}
\scalebox{0.8}{
\begin{tikzpicture}
\draw (0,0) -- (-5,-5) -- (5,-5) -- (0,0);
\draw[dashed] (0,0) -- (0.5,-1) -- (-0,-2) -- (0.5,-3) -- (0.5,-4) -- (0,-5);
\draw[dotted] (0,0) -- (0.5,-1) -- (-0,-2) -- (0.5,-3) -- (0.5,-3.75) -- (1,-4) -- (1.95,-5);
\draw (0.5,-1) -- (4.5,-5);
\draw (0.5,-1) -- (-3.5,-5);
\draw (-0,-2) -- (-3,-5);
\draw (0,-2) -- (3,-5);
\draw (0.5,-3) -- (-1.5,-5);
\draw (0.5,-3) -- (2.5,-5);
\draw (0.5,-4) -- (-0.5,-5);
\draw (0.5,-4) -- (1.5,-5);

\node at (0.5-0.3,-1) {$A$};
\node at (-0+0.3,-2) {$A$};
\node at (0.5-0.3,-3) {$A$};
\node at (0.5-0.3,-4) {$A$};

\draw (-5,-5) -- (-5,-6) -- (5,-6) -- (5,-5);
\draw[dashed] (-5,-5.5) -- (5,-5.5);

\node at (-4.25,-5.25) {$x$};
\draw[dotted] (-3.5,-5) -- (-3.5,-5.5);
\node at (-3.25,-5.25) {$y_1$};
\draw[dotted] (-3,-5) -- (-3,-5.5);
\node at (-2.25,-5.25) {$y_2$};
\draw[dotted] (-1.5,-5) -- (-1.5,-5.5);
\node at (-1,-5.25) {$y_3$};
\draw[dotted] (-0.5,-5) -- (-0.5,-5.5);
\node at (0.5,-5.25) {$w$};
\draw[dotted] (1.5,-5) -- (1.5,-5.5);
\node at (2,-5.25) {$\by_3$};
\draw[dotted] (2.5,-5) -- (2.5,-5.5);
\node at (2.75,-5.25) {$\by_2$};
\draw[dotted] (3,-5) -- (3,-5.5);
\node at (3.75,-5.25) {$\by_1$};
\draw[dotted] (4.5,-5) -- (4.5,-5.5);
\node at (4.75,-5.25) {$\bx$};

\draw (1.8,-6) -- (1.8,-5.5);
\draw (2.1,-6) -- (2.1,-5.5);
\node at (1.95,-5.75) {$b$};
\node at (-1.6,-5.75) {$\pi$};
\node at (3.55,-5.75) {$\alpha'$};

\end{tikzpicture}
}
\end{center}
\caption{Factorization of a witness $\alpha=(x,\bx)(y_1,\by_1)(y_2,\by_2)(y_3,\by_3)w=\pi b\alpha'$ w.r.t.\ a nonterminal $A$ occurring at least four times a long the dashed path in a derivation tree of $\alpha$ leading to a letter either within the lcp $\pi=\biglcp L$ or to the lcp-defining letter $b$ (the leaf of the dotted path).}
\label{fig:lcp-height}
\end{figure}

In fact, we will show the stronger claim that any path from the root of $t$ to any letter of $\pi b$ has length at most $3N$ (i.e.\ all the paths leading to the separating letter $b$ or a letter left of it, see Figure~\ref{fig:lcp-height}); note that any path that leads to a letter right of $b$ (i.e.\ into $\alpha'$) has to enter a subtree of height less than $N$ as soon as it leaves the path leading to $b$ because of the minimality of $\alpha$.
Hence, if all the paths leading to $b$ or a letter left of $b$ have length less than $3N$, the longest path in the derivation tree must have length at most $4N$.

So assume for the sake of contradiction that there is a path leading to a letter within $\pi b$ that has at least length $3N$ i.e.\ consists of at least $3N+1$ nonterminals. Then there is one nonterminal $A$ that occurs at least four times leading to a factorization
\[
\alpha = (x,\bar{x})(y_1,\bar{y}_1)(y_2,\bar{y}_2)(y_3,\bar{y}_3)w
\]
Note that $x\bx\neq \ew$, $y_i\by_i\neq \ew$ ($i = 1,2,3$), and $w\neq \ew$ as $G$ is proper.
As this path ends at $b$ or left of it, we have $xy_1y_2y_3 \noq \pi$. With $(x,\bx)(y_i,\by_i)(y_j,\by_j)w \in L$ for any $i,j\in\{1,2,3\}$ we thus obtain that $xy_iy_j\noq \pi$ and $xy_jy_i \noq \pi$ and thus $y_iy_j=y_jy_i$ for all $i,j\in\{1,2,3\}$. So $y_i = p^{k_i}$ for the same primitive $p$ using Lemma~\ref{lem:com}.

Let $L' = (x,\bar{x})[(y_1,\bar{y}_1)+(y_2,\bar{y}_2)+(y_3,\bar{y}_3)]^\ast w$ so that $\{xw\bx,\alpha\}\subseteq L'$.
By construction $L'\subseteq L$ and thus $\biglcp L \noq \biglcp L' \noq xw\bx \lcp \alpha$.
As $xw\bar{x}$ is shorter than $\alpha$, it cannot be a witness, so $\pi a \noq xw\bar{x}$ and $\pi = xw\bar{x}\lcp \alpha$. Hence
\[
\biglcp L = \sigma\lcp \alpha = \pi = xw\bar{x}\lcp \alpha \noqr \biglcp L' \noqr \biglcp L \ie \biglcp L = \biglcp L'
\]
It therefore suffices to consider $L'$ in the following; in particular, $\alpha$ has to be a witness w.r.t.\ $xw\bx$ of minimal length, too. (From here on, {\em witness} will always be w.r.t.\ $xw\bx$.)
By virtue of Theorem~\ref{thm:lcp-pump} we have
$\biglcp L' = \biglcp (x,\bar{x})[(y_1,\bar{y}_1)^{\le 2}+(y_2,\bar{y}_2)^{\le 2}+(y_3,\bar{y}_3)^{\le 2}] w$.
Note that $\biglcp L' \no xw\bar{x} \lcp xy_iw\bar{y}_i\bar{x}$ for any $i = 1,2,3$
as $\abs{xy_iw\bar{y}_i\bar{x}}< \abs{\alpha}$ and thus $xy_iw\by_i\bx$ cannot be a witness by minimality of $\alpha$.
So for some $I\in\{1,2,3\}$
\[
\biglcp L' = xw\bar{x}\lcp xy_Iy_Iw\bar{y}_I\bar{y}_I\bar{x} \noq \alpha
\]
i.e.\ $xy_Iy_Iw\bar{y}_I\bar{y}_I\bar{x}$ has to be also a witness.
Set $(y,\bar{y}):=(y_I,\bar{y}_I)$ and 
$L''=(x,\bar{x})(y,\bar{y})^\ast w$
so that $L''\subseteq L' \subseteq L$ and $\biglcp L = \biglcp L' = \biglcp L''$ as
\[
xw\bar{x}\lcp xyyw\bar{y}\bar{y}\bar{x} = \biglcp L \noq \biglcp L' \noq \biglcp L'' \noq xw\bar{x} \lcp xyyw\bar{y}\bar{y}\bar{x} \no xyw\by\bx
\]
As $xyw\by\bx$ is not a witness, Theorem~\ref{thm:lcp-pump} tells us that there is some $q$ satisfying
\[
yw=wq^k \wedge 
q\by \neq \by q \wedge
\biglcp L = \biglcp L'' = xw\bar{x} \lcp xywq\bar{y}\bar{x} 
\wedge xwq^k(\by\lcp q^\omega)\noq \biglcp L \no xw q^{k+1} (\by\lcp q^\omega)
\]
From this, we obtain:
\textbf{1.} As we already know that $y_i = p^{k_i}$ (as they commute), it follows that $p$ and $q$ are conjugates with $pw=qw$ s.t.\ $y_iw= wq^{k_i}$.
\textbf{2.} As $xwq^k\noq \biglcp L \no xwq^\omega$, we find some $m\ge 0$ and $\dq \no q$ s.t.\ $\pi =\biglcp L = xwq^kq^m\dq$ and, thus, $\pi a = xwq^kq^m\dq a \noq xw\bx$ and $\pi b = xwq^kq^m\dq b\noq xyyw\by\by \bx$. (Here, $b$ might change, yet it cannot become $a$ as $xyyw\by\by\bx$ is a witness.)
Additionally, from $\pi = xw\bx \lcp xyyw\by\by\bx \no xwq^{k+1} (\by\lcp q^\omega)$ we obtain $\pi c \noq xwq^{k+1}(\by \lcp q^\omega)$, i.e.\ $q^m \dq c \noq q\by \lcp q^\omega \no q^\omega$ and thus $\dq c \noq q$.
Hence, any word with prefix $xwq^{k+1}(\by \lcp q^\omega)$ is a witness.

\noindent 
If there was at least one $j\in\{1,2,3\}\setminus\{I\}$ with $k_j>0$ s.t.\ $y_j = p^{k_j} \neq \ew$, then  $(x,\bx)(y_j,\by_j)(y,\by)w$ would be a witness shorter than $\alpha$ as $y_j$ would give us at least one copy of $q$:
\[
\begin{array}{lcl@{\hspace{3mm}}l}
(x,\bx)(y_j,\by_j)(y,\by)w  & = & xy_j y w\by\by_j \bx\\
& \noqr & xwq^{k+k_j}\by & \text{(as $yw=wq^k$ and $y_jw=wq^{k_j}$)}\\ 
& \noqr & xwq^{k+k_j}(\by\lcp q^\omega) &\\ 
& \noqr & xwq^{k+1}(\by\lcp q^\omega) & \text{(as $k_j>0$ and $q^{k+1}(\by\lcp q^\omega)\no q^\omega$)}\\ 
\end{array}
\]
So for all remaining $j\in\{1,2,3\}\setminus\{I\}$ we have $y_j=\ew$ and thus
$\by_j \neq \ew$ as $G$ is proper and thus $y_j\by_j\neq \ew$. 
By Lemma~\ref{lem:lcp-reg-1} $\biglcp xw\by_j^\ast \bx = xw\bx \lcp xw\by_j \bx$, hence 
$\pi a \noq xw\by_j^\ast \bx$, i.e.\ $q^{k+m}\dq a \noq \by_j^\omega$.
If  $q^{m}\dq b\noq \by_j$ for some $j\in\{1,2,3\}\setminus\{I\}$ (recall $\dq b \noq q$), then as $a\neq b$
\[
xw\bx\lcp (x,\bx)(y,\by)(y_j,\by_j)w \stackrel{\text{(as $y_j=\ew$)}}= xw(\bx\lcp q^k\by_j\by\bx) = xw(q^{k+m}\dq a\lcp q^{k+m}\dq b) = \pi
\]
i.e.\ $xyy_jw\by_j\by\bx$ would be a shorter witness than $\alpha$.
Hence $\by_j \noq q^{m}\dq \no q^{k+m}\dq a$ for both $j\in\{1,2,3\}\setminus\{I\}$. Thus:
\[
\abs{q^\omega\lcp \by_j^\omega} \ge
\abs{q^{k+m}\dq} \ge
\abs{q} + \abs{q^m\dq} > 
\abs{q}+\abs{\by_j} - \gcd(\abs{q},\abs{\by_j}) 
\]
By the periodicity lemma of Fine and Wilf (Lemma~\ref{lem:fw}) this implies $\by_j = q^{k'_j}$ for some $k'_j>0$ (as $q$ primitive), and, subsequently as the final contradiction, that $xy_Iy_jw\bar{y}_j\bar{y}_I\bar{x}$ would be a shorter witness.
\end{proof}

\section{Small Equivalent Subsets of Languages}\label{sec_rep}

In this section we formally introduce a notion of equivalence of languages w.r.t.\ longest common prefixes.
The first main result of this section is that every non-empty language has an equivalent subset consisting of at most three elements.
In case of acyclic context-free languages, such a subset can be computed in polynomial time.
In combination with Theorem \ref{thm:witness-height}, we can lift the restriction on acyclicity.
This enables us to ultimately conclude that the longest common prefix of a context-free language
can be computed in polynomial time.

\begin{definition}\label{d:equiv}
Two languages $L,L'$ are \emph{equivalent w.r.t the \txtLCP (short: $L\equiv L'$)} iff
$\biglcp(Lw)=\biglcp(L'w)$ for all words $w\in \als$.
\end{definition}

\noindent
We observe that $L$ is equivalent to $L'$ w.r.t.\ the \txtLCP also after
union or concatenation from the left or right with arbitrary other languages.
Formally, this amounts to the following properties:

\begin{lemma}\label{equiv}
For all non-empty languages $L,L',\hat{L}$ with $L\equiv L'$ we have:
\[
\textbf{\textsf{1.}}\, \biglcp (L \hat L) = \biglcp (L' \hat L) \quad
\textbf{\textsf{2.}}\, \biglcp (\hat L L) = \biglcp (\hat L L') \quad
\textbf{\textsf{3.}}\, \biglcp (L \cup\hat L) = \biglcp (L' \cup\hat L)
\]
\end{lemma}
\begin{proof}
 The argument is as follows:
\begin{enumerate}
 \item $\begin{array}{lllllll}
\biglcp (L \hat L) &=&
\biglcp_{w\in\hat L}(\biglcp (L w)) &=&
\biglcp_{w\in\hat L}(\biglcp (L' w)) &=&
\biglcp (L' \hat L);
\end{array}$

\item $\begin{array}{lllllll}
 \biglcp (\hat{L} L) &=& \biglcp (\hat{L} (\biglcp L)) &=& \biglcp (\hat{L} (\biglcp L')) &=& \biglcp (\hat{L} L');
\end{array}$

\item $\begin{array}{lllllll}
 \biglcp (L \cup \hat{L}) &=& \biglcp L \lcp \biglcp \hat{L} &=& \biglcp L' \lcp \biglcp \hat{L} 
 &=& \biglcp (L' \cup \hat {L}).
\end{array}$
\end{enumerate}
\end{proof}

\noindent
The next lemma gives us an explicit formula for $\biglcp (Lw)$ for the special case
of the two-element language $L=\{u,uv\}$.
\begin{lemma}\label{l:omega}
Assume that $u,v\in\Sigma^*$ with $v\neq\epsilon$.
For all words $w \in \Sigma^*$, $\biglcp(\{u, uv\} w) = u (w\lcp v^\omega)$ holds.
\end{lemma}
\begin{proof}
 $\biglcp(\{u, uv\} w) = uw\lcp uvw$. If $w$ and $v$ are incomparable or
$w$ is a prefix of $v$, 
$w\lcp vw= w\lcp v = w\lcp v^\omega$, and the claim follows.
 Thus, it remains to consider the case that $v \sqsubseteq w$. Then $w = v^i w'$ for some $i$
so that $v$ is no longer a prefix of $w'$. Then
 $\biglcp(\{u, uv\} w) = \biglcp(\{u, uv\} v^iw') = 
uv^i(w'\lcp vw') = uv^i(w'\lcp v^\omega) = u(w\lcp v^\omega)$.
\end{proof}

\noindent
The explicit formula from Lemma \ref{l:omega} can be used to identify small equivalent sublanguages.

\begin{theorem}\label{t:three}
For every non-empty language $L\subseteq\Sigma^*$ there is a language $L' \subseteq L$ consisting of at most three words such that $L\equiv L'$.
\end{theorem}

\begin{proof}
If $L$ is a singleton language, we choose $L'=L$.
So assume that $L$ contains at least two words with \txtLCP $u$.
If the \txtLCP $u$ of $L$ is not contained in $L$ then we choose $L'$ as consisting of the two minimal words
$w_1, w_2$ so that $u = w_1 \lcp w_2$.
It remains to consider the case where the \txtLCP $u$ of $L$ is contained in $L$.
Then we have for each word $w\in\Sigma^*$,
\begin{equation}\label{eqn:lcpLw}
\begin{array}{llll}
\biglcp(Lw)	&=& \biglcp(\{uv\mid uv\in L\}w)	\\
		&=& \biglcp\{\biglcp(\{u,uv\}w)\mid uv\in L,v\neq\epsilon\}	\\
		&=& \biglcp\{u(w\lcp v^\omega)\mid uv\in L,v\neq\epsilon\}	& \text{(Lemma \ref{l:omega})}\\
		&=& u(w \lcp\biglcp\{v^\omega\mid uv\in L,v\neq\epsilon\})	\\
\end{array}
\end{equation}
If $L$ is ultimately periodic, then all words in $L$ are of the form $uv_0^i$ for some $v_0\in\Sigma^+$
and $i\geq 0$, and $(v_0^i)^\omega = v_0^\omega$. Thus, $\biglcp(Lw) = u(w \lcp v^\omega)$
for any $uv\in L$ with $v\neq\epsilon$. Hence,
$L\equiv L' = \{u,uv\}$ for any such $v$.

If $L$ is not ultimately periodic, then we choose
words $uv_1,uv_2\in L$ so that the \txtLCP of $v_1^\omega$ and $v_2^\omega$ has minimal length.
Then 
\[
\begin{array}{lll}
\biglcp(\{u,uv_1,uv_2\} w)		
		&=& u(w \lcp v_1^\omega\lcp v_2^\omega) 	\\
		&=& u(w \lcp\biglcp\{v^\omega\mid uv\in L,v\neq\epsilon\})	\\
\end{array}
\]
by the minimality of $v_1^\omega\lcp v_2\omega$. Therefore, $L \equiv L' = \{u,uv_1,uv_2\}$.
\end{proof}

\noindent
Since for any non-empty words $w_1,w_2$ given by SLPs, an SLP for  $w_1^\omega\lcp w_2^\omega = w_1w_2 \lcp w_2 w_1$ (if $w_1 \neq w_2$) can be computed in 
polynomial time\footnote{Lohrey \cite{Lohrey2012} gives an overview over the classical algorithms for SLPs. The fully compressed pattern matching problem for SLPs is in PTIME \cite[Theorem 12]{Lohrey2012}, i.e.\ we can test whether one SLP is a factor of another SLP. Especially we can test whether one SLP is a prefix of another SLP. As we can build an SLP
for any prefix of an SLP in polynomial time we can use a binary search to compute the \txtLCP of two SLPs in polynomial time.}, we have:
\begin{corollary}\label{c:finite}
 For every non-empty finite $L\subseteq\Sigma^*$ consisting of words each of which is represented by an SLP, a subset
 $L' \subseteq L$ consisting of at most three words can be calculated in polynomial time such that $L\equiv L'$.
\end{corollary}
\begin{proof}
 The proof distinguishes the same cases as in the proof of Theorem \ref{t:three} and relies on polynomial algorithms on SLPs \cite{Lohrey2012}.
 If $L$ contains at most three words we are done.
 Since the words in $L$ are given as SLPs, we can calculate (a SLP for) the \txtLCP $u$ of the words in $L$.
 Next, we determine whether $u$ is in $L$.
 This can again be checked in polynomial time. 
  If this is not the case, then we can select two words $w_1,w_2\in L$ so that $u=w_1\lcp w_2$
 giving us $L'=\{w_1,w_2\}$ in polynomial time.
  So, now assume that $u$ is in $L$.
 Next, we check whether or not $L$ is ultimately periodic, i.e., whether for any non-empty words $v_1,v_2$
 with $uv_1,uv_2\in L$, $v_1^\omega = v_2^\omega$.
 By Lemma~\ref{lem:fw} this is the case iff $v_1v_2=v_2v_1$. 
 The latter can be checked in polynomial time as concatenation and equality of SLPs can be calculated in polynomial time.
 If this is the case, then we obtain $L' =\{u,uv\}$ for some $uv\in L$ with $v\neq\epsilon$ in polynomial time.

 It remains to consider the case where the \txtLCP $u$ is contained in $L$ and $L$ is not ultimately periodic.
 Then we need to determine words $uv_1$ and $uv_2$ in $L$ with $v_1\neq\epsilon\neq v_2$ such that
 $v_1^\omega \lcp v_2^\omega$ has minimal length.
 Since $v_1^\omega \lcp v_2^\omega = v_1v_2\lcp v_2v_1$ (see Corollary~\ref{cor:fw}),
 such a pair can be computed in polynomial time as well. Therefore, $L'=\{u,uv_1,uv_2\}$
 can be computed in polynomial time.
\end{proof}

\noindent
The following lemma explains that equivalence of two non-empty languages of cardinalities at most $3$
can be decided in polynomial time.

\begin{lemma}\label{l:equiv}
 Let $L_1, L_2 \subseteq \Sigma^*$ denote non-empty languages consisting of at most three words each, which are all given
 by SLPs. Then $L_1\stackrel{?}{\equiv} L_2$ can be decided in polynomial time.
\end{lemma}
\begin{proof}
 If one of the two languages contains just a single word, 
 then $L_1 \equiv L_2$ iff $L_1 = L_2$ --- which can be
 decided in polynomial time.
  Otherwise, we first compute $\biglcp L_1$ and $\biglcp L_2$.
 If these differ, then by definition $L_1$ cannot be equivalent to $L_2$. 
  Therefore assume now that $u=\biglcp L_1 = \biglcp L_2$ is the common \txtLCP.
 
 Obviously, $L_i$ and $L_i\cup\{u\}$ are equivalent w.r.t.\ the \txtLCP ($i=1,2$).
 Thus, for testing equality, we may add $u$ to $L_1$ resp.\ $L_2$, if it is missing, and reduce $L_1$ resp.\ $L_2$ subsequently to languages of at most three words.
  
 From Equation \ref{eqn:lcpLw} follows that $L_1 \equiv L_2$ if
 $\biglcp\{v_1^\omega \mid uv_1 \in L_1, v_1 \neq \epsilon\} = \biglcp \{v_2^\omega \mid uv_2 \in L_2, v_2 \neq \epsilon\}$.
 This is the case if either $v_1^\omega = v_2^\omega$ for all $uv_1 \in L_1$ and $uv_2 \in L_2$ or
 for $uv_i, uv'_i \in L_i$, $v_i\neq\epsilon\neq v'_i$ with $w_i = v_i^\omega\lcp {v'}_i^\omega$ is minimal for $L_i$ ($i =1,2$),
 $w_1 = w_2$ holds.
 
 In the first case $v_1^\omega = v_2^\omega$ for all $uv_1 \in L_1$ and $uv_2 \in L_2$
 can be checked in polynomial time according to the periodicity lemma of Fine and Wilf
 (cf. Corollary \ref{cor:fw}). 
  In the second case $w_1,w_2$ can be computed and compared in polynomial time as all words are given as SLPs.
 Thus, we ultimately arrive at a polynomial time decision procedure. 
             \end{proof}
\begin{remark}
Note that in light of the equivalence test, we can choose distinct letters $a,b\in\Sigma$,
and equivalently replace the language $L_1=\{uv_1,uv_2\}$ with 
$L'_1=\{ua,ub\}$ whenever $v_1\neq\epsilon\neq v_2$ and $v_1\lcp v_2=\epsilon$, and
the language $L_2=\{u,uv_1,uv_2\}$ by the language $L'_2=\{u,uwa,uwb\}$ 
whenever $w =v_1v_2\lcp v_2v_1 \neq v_1v_2$ holds.
This reduced representation allows for an easier computation.
\end{remark}
\noindent
Now we have all pre-requisites to prove the main theorem of our paper.

\begin{theorem}\label{t:main}
Assume that $G$ is a proper context-free grammar with $L=L(G)$ non-empty.
Then the longest common prefix of $L$ can be calculated in polynomial time.
\end{theorem}
\begin{proof}
Assume w.l.o.g.\ that $G$ is a CFG in Chomsky normal form as this simplifies the notation. For the actual fixed-point iteration this is not required. 
Then we calculate $\biglcp L(G)$ as follows.
We build (implicitly, see the following remark) an acyclic CFG $\hat{G}$ in polynomial time such that $L(\hat{G})$ consists of all words
of $L(G)$ for which there is a derivation tree of height at most $4N$ where $N$ is the number
of nonterminals in $G$. 
To this end, we tag the variables with a counter that bounds the height of the derivation trees. In more detail, for every rewriting rule $A\to BC$ of $G$ and every $i\in\{1,\ldots,4N\}$ we add to $\hat{G}$ the rule $\bl{A}{i}\to \bl{B}{i-1}\bl{C}{i-1}$, and for every rule $A\to a$ of $G$ and every $i\in\{0,1,\ldots,4N\}$ we add the rule $\bl{A}{i}\to a$ to $\hat{G}$. In a derivation tree w.r.t.\ $\hat{G}$ every path starting at some node labeled by $\bl{A}{i}$ has thus length at most $i$ as $i$ has strictly decreases when moving down to towards the leaves, hence, a node labeled by $\bl{A}{i}$ can only be the root of a (sub-)tree of height at most $i$. Further, every derivation tree of $\hat{G}$ becomes a derivation tree of $G$ by simply replacing $\bl{A}{i}$ by $A$. As every rule of $G$ is copied at most $4N+1$ times with $N$ the number of nonterminals of $G$, the size of $\hat{G}$ grows at most quadratically with the size of $G$. In particular, $\hat{G}$ is still proper and in CNF.
For more details, see e.g.\ section 3 in \cite{DBLP:conf/calco/EsparzaL11}.

By Theorem \ref{thm:witness-height}, we know that $\biglcp L(G) = \biglcp L(\hat{G})$.
By construction, $\hat{G}$ is also in Chomsky normal form.
For $i$ from $0$ to (at most) $4N$ (with $N$ still the number of variables of the original grammar $G$ -- as $\hat{G}$ is acyclic we only need to compute $[\bl{A}{i}]$ once when proceeding bottom-up), we then compute in every iteration for every nonterminal $\bl{A}{i}$ (for the currently value of $i$) first the language
\begin{eqnarray*}
[\bl{A}{i}]':=\{a\in\Sigma^*\mid \bl{A}{0}\to a\in P\}\cup
\bigcup_{A\to BC \in G}[\bl{B}{i-1}]\cdot[\bl{C}{i-1}]		\label{eq:set}
\end{eqnarray*}
By induction on $i$, we may assume that the languages $[\bl{B}{i-1}],[\bl{C}{i-1}]$ (a) have already been computed, (b) consist of at most three words, and (c) every word is given as an SLP.
Note that the cardinality of every language $[\bl{A}{i}]'$ is polynomial in the size of $G$.
By virtue of Corollary \ref{c:finite}, we therefore can reduce $[\bl{A}{i}]'$ in polynomial time to a language $[\bl{A}{i}]\subseteq [\bl{A}{i}]'$ with $[\bl{A}{i}]\equiv [\bl{A}{i}]'$ and $\abs{[\bl{A}{i}]}\le 3$.
By construction, we then have
\begin{eqnarray*}
[\bl{A}{i}] &\equiv&	\{w\in\Sigma^*\mid \bl{A}{i}\Rightarrow^* w\}	\label{eq:A_rep}
\end{eqnarray*}
Since $\hat G$ has polynomially many nonterminals only, the overall algorithm runs in 
polynomial time.
\end{proof}
\begin{remark}
Note that we can drop the assumption that the grammars $G$ and likewise $\hat G$ are in Chomsky normal 
form
if the right-hand sides of all rules have bounded lengths. 
Then the cardinality of the languages $[\bl{A}{i}]'$ are still polynomial.
Further, instead of spelling out the grammar $\hat G$ explicitly, we may perform a
round robin fixpoint iteration where in every round we first compute $$[A]':=\bigcup_{A\to w_1B_1w_2 B_2\ldots w_k B_k w_{k+1}} \{w_1\}\cdot [B_1]\cdot\{w_2\}\cdot [B_2] \cdots \{w_k\}\cdot [B_k]\cdot \{w_{k+1}\}$$ with initially $[A]:=\{w\in\als\mid A\to w \in G\}$, then updating $[A]$ so that $[A]\subseteq [A]'$ with $[A]\equiv [A]'$ and $\abs{[A]}\le 3$. Theorem \ref{thm:witness-height} guarantees that the \txtLCP is attained after at most $4N$ iterations. Using standard approaches like work lists, we only need to recompute $[A]$ if there is some rule $A\to \gamma B\delta$ in $G$ and $[B]$ has changed since the last recomputation of $[A]$. As shown in Lemma~\ref{l:equiv} we can easily check if $[B]\not\equiv [B]'$ in every round and accordingly insert $A$ into the work list.
\end{remark}
\noindent
We demonstrate this simplified version of the algorithm described in Theorem \ref{t:main} by an example.

\begin{example}\label{ex:alg}
Consider the following grammar $G$ with the following rules:
\begin{align*}
S &\rightarrow A ababaac \\
A &\rightarrow ab\,A\,abaab \mid ab\,A\,abaac \mid \epsilon
\end{align*}
  The round robin fixpoint iteration would proceed by iteratively evaluating the equations
 \[
 \begin{array}{llcl}
 & [A]' & := & \{ ab w abaab, abwabaac, \ew \mid w \in [A]\}\\
 & [S]' & := & \{wababaac \mid w\in [A]\}\\
 \end{array}
 \]
 and recomputing the languages $[A]$ and $[S]$ so that $[A]\equiv [A]'$ and $[S]\equiv [S]'$ and both $[A]$ and $[S]$ consist of at most three words where we further reduce the words of $[A]$ and $[S]$ as described in the remark following Lemma~\ref{l:equiv}.
 As $[A]$ does not depend on $[S]$, we can postpone the computation of $[S]$ after $[A]$ has converged.
  In the first round, we have:
 \[
 [A] = [A]' = \{\epsilon\}
 \]
 For the second round, we first calculate:
 \[
 [A]' = ab \{\epsilon\} abaab \cup ab \{\epsilon\} abaac \cup \{\epsilon\} = \{ababaab, ababaac, \epsilon\}
 \]
 and thus update $[A]$ to $[A] := \{(ab)^2aab, (ab)^2aac, \epsilon\}$.
 For the third round, we obtain
  \[
  \begin{array}{lll}
  [A]' & = & ab \{(ab)^2aab, (ab)^2aac, \epsilon\} abaab \cup ab \{(ab)^2aab, (ab)^2aac, \epsilon\} ab aac \cup \{\epsilon\} \\
  &=& \{(ab)^3a(ab)^2aab, (ab)^3aacabaab, (ab)^2aab\} \cup \\
  && \{(ab)^3a(ab)^2aac, (ab)^3aacabaac, (ab)^2aac\} \cup \{\epsilon\}\\
  & \equiv& \{(ab)^3aababaab, (ab)^2aab, \epsilon\} \\
  & \equiv& \{(ab)^3, (ab)^2aa, \epsilon\}\\   & =: & [A]
  \end{array}
  \]
  which is already the fixpoint.
  Therefore we obtain
  \[
   \begin{array}{lll}
   [S]' & =& \{(ab)^3, (ab)^2aa, \epsilon\} ababaac\\
    &=& \{(ab)^3(ab)^2aac, (ab)^2aa(ab)^2aac, (ab)^2aac\}\\
  	   &\equiv& \{(ab)^3(ab)^2aac, (ab)^2aac\} \\
  	   &\equiv& \{(ab)^3, (ab)^2aa\}\\   	   & =: & [S]
   \end{array}
   \]
   So $\biglcp L = (ab)^3\lcp (ab)^2aa = (ab)^2a$.
\end{example}

\section{Conclusion}\label{s:conclusion}

We have shown that the longest common prefix of a non-empty context-free language
can be computed in polynomial time. This result was based on two structural results, namely, that
it suffices to consider words with derivation trees of bounded height, and second that each non-empty language
is equivalent to a sublanguage consisting of at most three elements.
For the actual algorithm, we relied on succinct representations of long words by means of SLPs.
It remains as an intriguing open question whether the presented method can be generalized to more 
expressive grammar formalisms.

\section{Appendix}

\subsection{Proof of Theorem~\ref{thm:lcp-pump}}

We split the proof of the theorem into several lemmata covering the cases
\begin{enumerate}
\item
$L=(x,\ew)[(y,\ew)+(z,\ew)]^\ast w=x(y+z)^\ast w$ (cf.\ Lemma~\ref{ssec:lcp-reg})
\item
$L=(x,\bx)(y,\by)^\ast w$ (cf.\ Lemma~\ref{ssec:lcp-one-pt})
\item
$L=(x,\bx)[(y,\by)+(z,\bz)]^\ast w$ (cf.\ Lemma~\ref{ssec:lcp-two-pt}), and
\item
$L=(x,\bx)[(y_1,\by_1)+\ldots+(y_l,\by_l)]^\ast w$ for arbitrary $l\in\N$. (cf.\ Lemma~\ref{ssec:lcp-mult-pt}).
\end{enumerate}

\begin{lemma}\label{ssec:lcp-reg}
$L=x(y+z)^\ast w \Rightarrow \biglcp L = x\biglcp (y+z)^{\le 1}w$
\end{lemma}
\begin{proof}
As $x$ does not matter, simply assume $x=\ew$.
We show by induction on $m$ that for any $\alpha \in (y+z)^m w$
\[
w\lcp yw \lcp zw = w\lcp yw \lcp zw \lcp \alpha
\]
The case $m \le 1$ is obviously true. Fix any $m> 1$ and any $\alpha \in (y+z)^{m+1}w$; wlog.\ $\alpha = \alpha' yw$. Set $w' = w\lcp yw$. Then:
\[
  w\lcp yw \lcp zw \stackrel{\text{Induction}}{=} w \lcp yw \lcp zw \lcp \alpha' w =  w' \lcp zw \lcp \alpha' w' = w' \lcp zw \lcp \alpha'(w\lcp yw) = w\lcp yw \lcp zx \lcp \alpha' w \lcp \alpha' y w
\]
\end{proof}

\begin{lemma}\label{ssec:lcp-two-pt}
Let $L= (x,\bar{x})[(y,\bar{y})+(z,\bar{z})]^\ast w$. Then $\biglcp L = \biglcp (x,\bar{x})[ (y,\bar{y})^{\le 2} + (z,\bar{z})^{\le 2}]w$.
\end{lemma}

\begin{proof}
The case $y=\ew=z$, i.e.\ $L = xw(\bar{y}+\bar{z})^\ast \bar{x}$ is already proven in Lemma~\ref{ssec:lcp-reg}.

Consider the case $yz\neq zy$. Let $\alpha$ be a witness (w.r.t.\ $xw\bx$). 
Assume $\alpha$ is of the following form for some suitable $j\ge 0$
\[
\alpha = xy^j y z \alpha' \bz\by\by^j \bx
\]
(The case $\alpha = xz^j zy\alpha'\by\bz\bz^j\bx$ is symmetrical.)
Then (swapping the inner most $y$ and $z$ still yields a word of $L$):
\[
\biglcp L = xw\bx \lcp \alpha = xw\bx \lcp xy^jyz\alpha' \bz\by\by^j \bx \lcp xy^jzy \alpha' \by\bz\by^j \bx = xw\bx \lcp xy^j(yz\lcp zy)
\]
Using Corollary~\ref{cor:fw}:
\[
xw\bx \lcp xy^j(yz\lcp zy) = xw\bx \lcp xy^j(y^\omega\lcp z^\omega)
\]
But obviously
\[
xyzw\bz\by\bx \lcp xzyw\by\bz\bx = x(yz\lcp zy)\noq xy^j(yz\lcp zy)
\]
So either $xyzw\bz\by\bx$ or $xzyw\by\bz\bx$ has to be a witness, too.

But again by virtue of Corollary~\ref{cor:fw} and for sufficiently large $j$
\[
xyzw\bz\by\bx \lcp xzyw\by\bz\bx = xy^\omega \lcp xz^\omega  = \px\py^jw \lcp \px\pz^jw
\]
Hence, we already find a witness within $\px[\py^\ast +\pz^\ast]w$ and, thus, within $\px[\py^{\le 2}+\pz^{\le 2}]w$

So assume for the following that $yz=zy$ with $y=p^k\wedge z=p^l$ and $p$ primitive (wlog.\ $k\ge l$). Wlog.\ $\max\{k,l\}>0$. (If $k=0$, then $y=\ew=z$ which we have already discussed.)

Let $w=p^mp'w'$ with $p=p'p''$ and $w\lcp p^\omega = w \lcp p^{m+1} = p^m p' \no p^{m+1}$.
Set $q=p''p'$ s.t.\ $pp'=p'q$ and $q\lcp w'=\ew$ and $p''\neq\ew$. As $p$ is primitive, so is $q$.

If $w\not\no p^\omega$, then $w'\neq \ew$. If $y\neq\ew$, then $xyw\by\bx$ is a witness; if $z\neq\ew$, then $xzw\bz\bx$ is a witness, too.

Hence, $w\no p^\omega$ in the following. Then $w'=\ew$ and $pw=wq$.
We factorize $\bx,\by,\bz$ w.r.t.\ $q$:
\begin{itemize}
\item
Let $\bx\lcp q^\omega = q^nq'\no q^{n+1}$ with $\bx= q^n q' \bx'$ and $q=q'q''$ and $q''\neq\ew$.
\item
Let $\bar{y}\lcp q^\omega = q^{k'}\hat{q}\no q^{k'+1}$ with $\bar{y}=q^{k'}\hat{q}\bar{y}'$ and $q=\hat{q}\hat{\hat{q}}$ and $\hat{\hat{q}}\neq\ew$.
\item
Let $\bar{z}\lcp q^\omega = q^{l'}\dot{q}\no q^{l'+1}$ with $\bar{z}=q^{l'}\dot{q}\bar{z}'$ and $q=\dot{q}\ddot{q}$ and $\ddot{q}\neq\ew$.
\end{itemize}
Thus:
\[
L = (xw,q^nq'\bx')[(q^k,q^{k'}\hq\by')+(q^l,q^{l'}\dq\bz')]^\ast \ew
\]
Wlog.\ $y = p^k \neq \ew$, i.e.\ $k>0$, and further $k\ge l$.

Let $\alpha$ be a witness w.r.t.\ $xw\bx$. We may distinguish the following cases for a witness $\alpha$ (with $\Gamma = (y,\by)+(z,\bz)$):
\begin{enumerate}
\item[(1)] $\alpha \in (x,\bx)(y,\by)^\ast w$
\item[(2)] $\alpha \in (x,\bx)(z,\bz)^\ast w$
\item[(3)] $\alpha \in (x,\bx)\Gamma^\ast (z,\bz)(y,\by)^+ w$
\item[(4)] $\alpha \in (x,\bx)\Gamma^\ast (y,\by)(z,\bz)^+ w$
\end{enumerate}
(The case $\alpha = xw\bx$ is covered by both (1) and (2).)

Cases (1) and (2) are both covered by Lemma~\ref{ssec:lcp-one-pt}.
Hence, we may assume in the following that any witness is of the form (3) or (4) --- otherwise we are done.

As $k>0$, we have
\[
\biglcp L \noq xw\bx\lcp (x,\bx)(y,\by)^{n+1}w = xw(\bx\lcp q^{k(n+1)}) = xw(\bx\lcp q^{n+1}) = xwq^n q'
\]
Hence 
$\biglcp L = xw\phi \text{ for some suitable } \phi \noq q^nq'$;
if $\phi = q^nq'$, then $(x,\bx)(y,\by)^{n+1}w$ would be a witness, contradicting our assumption that any witness of the form (3) or (4).

Hence, $\phi \no q^nq'$, i.e.\ $\biglcp L = xw\phi \no xwq^nq' = xw(\bx \lcp q^\omega)$; so, any witness $\alpha$ has to satisfy
\[
xw\phi = xw\bx\lcp \alpha = xwq^nq' \lcp \alpha = xwq^\omega \lcp \alpha
\]
as $\alpha$ has to differ from $xw\bx$ within the suffix $q^nq'$.

If $n < k$, then
\[
xw\bx \lcp \underbrace{x \ldots y \ldots w \ldots \by \ldots \bx}_{\in (x,\bx)\Gamma^\ast (y,\by)\Gamma^\ast w} = xw(\bx \lcp q^k) = xw(\bx \lcp q^{n+1}) = xwq^nq'
\]
So, either $xyw\by\bx$ is a witness, or $y$ cannot occur in any witness, implying that either $xzw\bz\bx$ or $xzzw\bz\bz\bx$ is a witness.

Hence, $n\ge k \ge l$ with $n>0$ as $k>0$.

We need to take a closer look at the structure of a respective $\alpha$. As case (4) is a special case of (3), we discuss (3) in detail and only remark where the proof differs from case (4).

So assume $\alpha \in (x,\bx)\Gamma^\ast (z,\bz)(y,\by)^+ w$. Then
	\[
	\alpha = xw q^{\lambda l + \mu k} q^{l} q^{jk} q^k q^{k'} \hq \by' (q^{k'}\hq\by')^j q^{l'}\dq\bz' \beta q^nq'\bx'
	\]
	where $\lambda$ ($\mu$) is the number of $z$ ($y$) right of $x$ and left of the inner most $z$; and $\beta$ is the corresponding string of $\by,\bz$, e.g.\ if $\lambda+\mu=0$, then $\beta=\ew$.
	
	As $\alpha \lcp xw\bx = xw\phi \no xwq^nq'$, we have 
	\[
	q^{k}q^{k'}\hq \noq q^l q^{k}q^{k'}\hq \noq q^{\lambda l + \mu k} q^{l} q^{jk} q^k q^{k'} \hq \noq \phi \no q^n q'
	\]
	If $\by'\neq \ew$, then
	\[
	xw\phi = \alpha \lcp xwq^\omega = xwq^{\lambda l + \mu k} q^{l} q^{jk} q^k q^{k'} \hq \no xw\bx \lcp xyw\by\bx = xw(q^nq'\lcp q^{k}q^{k'}\hq) = xwq^kq^{k'}\hq
	\]
	So $\by'=\ew$ in the following.
    
    We first do away with the case $\hq =\ew$:
    \begin{adjustwidth}{0.75cm}{0cm}
    If $\hq= \ew$ (i.e.\ $\by=q^{k'}$), then
	\[
	\alpha = xw q^{\lambda l + \mu k} q^{l} q^{jk} q^k q^{k'} (q^{k'})^j q^{l'}\dq\bz' \beta q^nq'\bx'
	\]
	Hence:
	\[
		q^l q^{l'}\dq \noq q^l q^{l'} \dq \noq q^{\lambda l + \mu k}q^{l} q^{jk} q^k q^{k'} (q^{k'})^j q^{l'}\dq \noq \phi \no q^nq'
	\]
	{\em Note}\,: As $k>0$, we trivially have $q q^l q^{l'} \dq \noq q^{\lambda l + \mu k} q^{l} q^{jk} q^k q^{k'} (q^{k'})^j q^{l'}\dq$ in contrast to case (4); but we always may assume in the case $\hq\by'=\ew$ that $k+k'>0$ as otherwise $(y,\by)=\ew$ which also holds analogously in case (4).

    Obviously $\dq\bz'\neq \ew$, as otherwise $xw\bx\lcp \alpha = xwq^nq'$.

	If $\bz'\neq \ew$, we obtain the contradiction $xw\bx \lcp \alpha \no xw\bx \lcp xzw\bz\bx = xwq^lq^{l'}\dq$ as $q^{l+l'}\dq \no q^nq'$; so $\dq\neq \ew = \bz'$ (i.e.\ $\bz=q^{l'}\dq$), and $\beta \in (q^\ast \dq)^\ast$.
	
	As $q$ primitive and $\ew\neq \dq \no q$, we have $\dq q\neq q \dq$ and thus $\dq^\omega\lcp q^\omega = \dq q \lcp q\dq \no q\dq$.

	Hence (using $n>0$ and $qq^{l+l'}\dq \no q^nq'$):
	\[
	xw\phi \no xw\bx \lcp xzw\bz\bx = xw(q^nq' \lcp q^{l+l'}\dq q) = xwq^{l+l'}(q\dq \lcp \dq q) \no xwq^{l+l'}q\dq \noq xw\phi
	\]
    \end{adjustwidth}    
    Thus also $\hq\neq \ew$ from here on.
	
	Again, as $q$ primitive and $\ew\neq \hq \no q$, we have $\hq q\neq q \hq$ and thus $\hq^\omega\lcp q^\omega = \hq q \lcp q\hq \no q\hq$.

	Hence (using $n>0$)
	\[
	xw\phi \no xw\bx \lcp xyw\by\bx = xw(q^nq' \lcp q^{k+k'}\hq q) \no xw qq^{k+k'}\hq
	\]
    If $\lambda l + \mu k+l+jk>0$, we obtain the contradiction
	\[
	qq^{k+k'}\hq \noq \phi \no qq^{k+k'}\hq
	\]
	analogously to the case $\hq\by'=\ew$.
	
    {\em Note}: In case (4) we are done at this point, as $k>0$ takes the place of $l\ge 0$ in case (4). 
    
	So $\lambda l + \mu k +l+jk=0$, i.e.\ $l=j=\mu = 0$ as $k>0$ for the following allowing us to write
	\[
	\alpha = xw q^k q^{k'} \hq q^{l'}\dq\bz' (q^{l'}\dq \bz')^\lambda q^nq'\bx'
	\]
    If $l'>0$, then (using $n>0$ and $\hq \no q$)
	\[
	xw\phi = \alpha \lcp xwq^\omega \no xyw\by\bx \lcp xwq^\omega = xw(q^{k+k'}\hq q\lcp q^\omega) = \alpha \lcp xwq^\omega
	\]
	So we have to have $l'=0$ which allows us to further simplify $\alpha$:
	\[
	\alpha = xw q^k q^{k'} \hq \dq\bz' (\dq \bz')^\lambda q^nq'\bx'
	\]
    If $\bz'\neq \ew$, then (using $n>0$ and $\dq \no q$)
	\[
	xw\phi \no xw\bx \lcp xzw\bz\bx = xw(q^nq'\lcp xw\dq\bz' q^nq') = xw\dq \no xwq \noq xwq^kq^{k'} \hq \noq xw\phi
	\]
	So $\bz'=\ew$ and thus $\dq \neq \ew$ (else $z=\ew$ and $\bz=\ew$):
	\[
	\alpha = xw q^k q^{k'} \hq \dq (\dq)^\lambda q^nq'\bx'
	\]
	Again we then have $\dq q \neq q\dq$, i.e.\ $q^\omega \lcp \dq^\omega = q\dq \lcp \dq q \no q\dq$.
	
	Hence (using $n>0$ and $\dq \no q$): $xw\bx \lcp xzw\bz\bx = xw(q^nq' \lcp \dq q^n q') \no xwq\dq$.
    
	We therefore have
    \[
    xw q^{k+k'}\hq \noq \alpha \lcp xwq^\omega = xw\phi \no xwq\dq
	\]
	i.e.\ $k'=0$, $k=1$, and $\hq \no \dq$ s.t.:
	\[
	\alpha = xw q\hq \dq (\dq)^\lambda q^nq'\bx'
	\]
	As $n>0$, $\hq\no\dq \no q$, $\phi \no q\dq$, we obtain the final contradiction:
	\[
	xw\phi=\alpha \lcp xw q^\omega \stackrel{\abs{\phi}<\abs{q\dq}\le \abs{q\hq\dq}}{=} xw(q\hq \dq \lcp q^\omega) \stackrel{\dq \no q,\abs{\phi}< \abs{q\dq}}{=} xw(q\hq q \lcp q^\omega) = xyw\by\bx \lcp xwq^\omega
	\]
\end{proof}

\begin{lemma}\label{ssec:lcp-mult-pt}
Let $L=(x,\bar{x})[\sum_{i=1}^n (y_i,\bar{y}_i)]^\ast w$. Then $\biglcp L = \biglcp (x,\bar{x})[\sum_{i=1}^n (y_i,\bar{y}_i)^{\le 2}] w$.
\end{lemma}
\begin{proof}
Let $\alpha\in L$ be a witness i.e.\ $\biglcp L =xw\bar{x}\lcp \alpha$.

Then $$\alpha = (x,\bar{x})\prod_{j=1}^k (y_{i_j},\bar{y}_{i_j}) w$$ for suitable $i_1,\ldots,i_k \in \{1,\ldots,n\}$.

If $k=1$, we are done.
Assume $k\ge 2$ for any witness (and any such factorization).
Pick a witness $\alpha$ and a factorization that minimizes $k$.
Set $$(y,\bar{y})=(y_{i_1},\bar{y}_{i_1}) \qquad (z,\bar{z})=(y_{i_2},\bar{y}_{i_2}) \qquad w' = \prod_{j=3}^k (y_{i_j},\bar{y}_{i_j}) w$$
Then using Lemma~\ref{ssec:lcp-two-pt}
\[
\biglcp L = xw\bar{x} \lcp \alpha = xw\bx\lcp  \biglcp \underbrace{(x,\bar{x})[ (y,\bar{y})+(z,\bar{z})]^\ast w'}_{\subseteq L\ni\alpha}
= xw\bar{x} \lcp \biglcp (x,\bar{x})[ (y,\bar{y})^{\le 2}+(z,\bar{z})^{\le 2}] w'
\]
So, one of the words on the right-hand side has to be a witness too. If $k=2$, we have $w'=w$, and we are also done. Hence assume $k\ge 3$ from now on.
Because of our assumption that $\alpha$ is a witness with a minimal factorization, only $(x,\bar{x})(y,\bar{y})^2 w'$ or $(x,\bar{x})(z,\bar{z})^2w'$ can be a witness.
Because of symmetry, it suffices to assume $\biglcp L = xw\bar{x}\lcp (x,\bar{x})(y,\bar{y})^2 w'$.
As $k\ge 3$ we have $w'=(y_{i_3},\bar{y}_{i_3})\prod_{j=4}^{k} (y_{i_j},\bar{y}_{i_j}) w$. Set $(z',\bar{z}')=(y_{i_3},\bar{y}_{i_3})$ and $w'' = \prod_{j=4}^{k} (y_{i_j},\bar{y}_{i_j}) w$ so that
\[
\biglcp L = xw\bar{x}\lcp (x,\bar{x})(y,\bar{y})^2 w' = xw\bar{x}\lcp (x,\bar{x})(y,\bar{y})^2 (z',\bar{z}')w''
= xw\bar{x} \lcp \biglcp \underbrace{(x,\bar{x})[(y,\bar{y})^2+(z',\bar{z}')]^\ast w''}_{\subseteq L} 
\]
Using again Lemma~\ref{ssec:lcp-two-pt} this is equivalent to
\[
\biglcp L =  xw\bar{x} \lcp xw''\bar{x}\lcp (x,\bar{x})(y,\bar{y})^2 w'' \lcp (x,\bar{x})(z',\bar{z}')w'' \lcp (x,\bar{x})(y,\bar{y})^4w'' \lcp (x,\bar{x})(z',\bar{z}')^2 w''
\]
As $(x,\bar{x})(y,\bar{y})^\ast w'\subseteq L$, adding $\biglcp (x,\bar{x})(y,\bar{y})^\ast w''$ cannot change the lcp 
\[
\biglcp L = xw\bar{x} \lcp xw''\bar{x}\lcp (x,\bar{x})(z',\bar{z}')w'' \lcp (x,\bar{x})(z',\bar{z}')^2 w'' \lcp \biglcp (x,\bar{x})(y,\bar{y})^\ast w''
\]
Using Lemma~\ref{ssec:lcp-one-pt} we finally obtain
\[
\biglcp L =  xw\bar{x} \lcp xw''\bar{x}\lcp (x,\bar{x})(z',\bar{z}')w'' \lcp  (x,\bar{x})(z',\bar{z}')^2 w'' \lcp \biglcp (x,\bar{x})(y,\bar{y})^{\le 2} w''
\]
Again, we have to find another witness within the words occurring on the right-hand side. But all these words have a factorization using less factors than $\alpha$ contradicting our choice of $\alpha$.
Hence, there has to be a witness having a factorization with $k\le 2$. Thus:
\[
\biglcp L = \biglcp_{i=1}^n (x,\bar{x})(y_i,\bar{y}_i)^{\le 2} w
\]
\end{proof}

\end{document}